\documentclass[11pt]{article}
\usepackage{stmaryrd}
\usepackage{bbm}
\usepackage{mathrsfs}
\usepackage{amsfonts}
\usepackage{amssymb}
\usepackage{amsmath}
\usepackage{amsthm}
\usepackage{graphics}
\usepackage{accents}
\usepackage{epic,color}
\usepackage{lscape,cite}
\usepackage[figuresright]{rotating}
\newtheorem{theorem}{Theorem}
\newtheorem{proposition}{Proposition}
\newtheorem{lemma}{Lemma}

\def \h#1{\widehat{#1}}
\def \t#1{\widetilde{#1}}
\def \b#1{\overline{#1}}
\def \ub#1{\underline{#1}}
\def \th#1{\widehat{\widetilde{#1}}}

\allowdisplaybreaks[4]
\begin{document}

\title{ On decomposition of the ABS lattice equations and related B\"{a}cklund transformations  }

\author{Danda Zhang,~~ Da-jun Zhang\footnote{Corresponding author. Email: djzhang@staff.shu.edu.cn}\\
{\small  Department of Mathematics, Shanghai University, Shanghai 200444, P.R. China}
}
\date{\today}
\maketitle

\begin{abstract}

The Adler-Bobenko-Suris (ABS) list contains all scalar quadrilateral equations which are consistent around the cube.
Each equation in the ABS list admits a beautiful decomposition.
In this paper, we first revisit these  decomposition formulas, by which we construct B\"{a}cklund transformations (BTs)
and consistent triplets.
Some BTs are used to construct new solutions, lattice equations and weak Lax pairs.

\vspace {10pt}

\end{abstract}

\section{Introduction}\label{sec-1}

It is well known that discrete integrable systems play important roles in the research of variety of areas such as statistic physics,
discrete differential geometry and discrete Painlev\'e theory.
Quadrilateral equations are partial difference equations defined by four points (see Fig.\ref{F:1}(a)), with a form
\begin{equation}\label{Q}
    Q(u,\t{u},\h{u},\h{\t{u}};p,q)=0,
\end{equation}
where $u$ is a function of discrete variables $n, m$, constants $p, q$ serve as spacing parameters of $n$-direction and $m$-direction, respectively.
We use short hand notations:
\begin{equation*}
   u=u_{n,m},~~ \t{u}=u_{n+1,m},~~ \h{u}=u_{n,m+1},~~ \th{u}=u_{n+1,m+1}.
\end{equation*}
In a beautiful paper \cite{ABS03} Adler Bobenko and Suris (ABS) classified all quadrilateral equations
with the assumption of consistency-around-the-cube: (i) affine linear w.r.t.  $u,\t{u},\h{u},\h{\t{u}}$;
(ii) equation \eqref{Q} having D4 symmetry;
(iii) tetrahedron property, i.e. embedding equation \eqref{Q} on the six faces of the cube in Fig.\ref{F:1}(b),
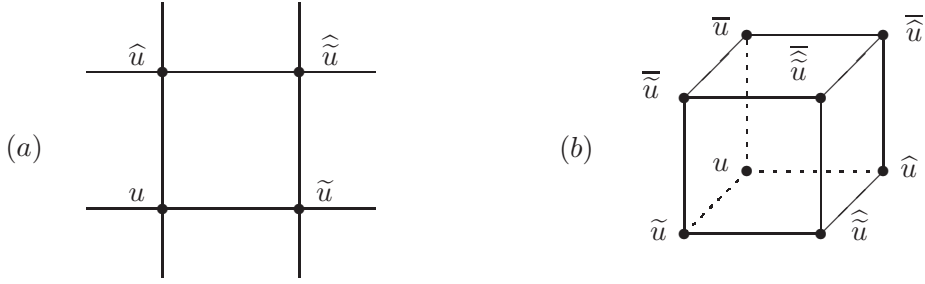
\begin{figure}[h]
\setlength{\unitlength}{0.0004in}
\hspace{2cm}
\begin{picture}(3482,2813)(0,-10)
\put(-800,1510){\makebox(0,0)[lb]{$(a)$}}
\put(1275,2708){\circle*{150}}
\put(825,2808){\makebox(0,0)[lb]{$\h u$}}
\put(3075,2708){\circle*{150}}
\put(3375,2808){\makebox(0,0)[lb]{$\h{\t u}$}}
\put(1275,908){\circle*{150}}
\put(825,1008){\makebox(0,0)[lb]{$u$}}
\put(3075,908){\circle*{150}}
\put(3300,1008){\makebox(0,0)[lb]{$\t u$}}
\drawline(275,2708)(4075,2708)
\drawline(3075,3633)(3075,0)
\drawline(275,908)(4075,908)
\drawline(1275,3633)(1275,0)
\end{picture}
\hspace{4cm}
\begin{picture}(3482,3700)(0,-500)
\put(-1200,1000){\makebox(0,0)[lb]{$(b)$}}
\put(450,1883){\circle*{150}}
\put(-100,1883){\makebox(0,0)[lb]{$\b{\t u}$}}
\put(1275,2708){\circle*{150}}
\put(825,2708){\makebox(0,0)[lb]{$\b u$}}
\put(3075,2708){\circle*{150}}
\put(3375,2633){\makebox(0,0)[lb]{$\b{\h u}$}}
\put(2250,83){\circle*{150}}
\put(2650,8){\makebox(0,0)[lb]{$\h{\t u}$}}
\put(1275,908){\circle*{150}}
\put(825,908){\makebox(0,0)[lb]{$u$}}
\put(2250,1883){\circle*{150}}
\put(1850,2108){\makebox(0,0)[lb]{$\b{\h{\t u}}$}}
\put(450,83){\circle*{150}}
\put(0,8){\makebox(0,0)[lb]{$\t u$}}
\put(3075,908){\circle*{150}}
\put(3300,833){\makebox(0,0)[lb]{$\h u$}}
\drawline(1275,2708)(3075,2708)
\drawline(1275,2708)(450,1883)
\drawline(450,1883)(450,83)
\drawline(3075,2708)(2250,1883)
\drawline(450,1883)(2250,1883)
\drawline(3075,2633)(3075,908)
\dashline{60.000}(1275,908)(450,83)
\dashline{60.000}(1275,908)(3075,908)
\drawline(2250,1883)(2250,83)
\drawline(450,83)(2250,83)
\drawline(3075,908)(2250,83)
\dashline{60.000}(1275,2633)(1275,908)
\end{picture}
\caption{(a). The points on which equation \eqref{Q} is defined. (b). the
  consistency cube.\label{F:1}}
\end{figure}
where shift for the third direction is denoted by bar with spacing parameter $r$ and the six equations are
\begin{eqnarray}
\begin{aligned}\label{BTCAC}
&&Q(u,\t{u},\h{u},\h{\t{u}};p,q)=0,~~Q(\b{u},\b{\t{u}},\b{\h{u}},\b{\h{\t{u}}};p,q)=0,\\
&&Q(u,\t{u},\b{u},\t{\b{u}};p,r)=0,~~Q(\h{u},\h{\t{u}},\h{\b{u}},\h{\t{\b{u}}};p,r)=0,\\
&&Q(u,\h{u},\b{u},\h{\b{u}};q,r)=0,~~Q(\t{u},\t{\h{u}},\t{\b{u}},\t{\h{\b{u}}};q,r)=0,
\end{aligned}\end{eqnarray}
with initial value $u, \t u, \h u, \b u$,
  the final value of $\th{\b u}$ only depends on $\t u, \h u, \b u$ and is same no mater
which equation is finally used to calculate it.
ABS list includes all the quadrilateral equations that are consistent around the cube (CAC):
\begin{subequations}\label{ABS}
\begin{align}
\text{H1:} \ \ & (u-\h{\t{u}})(\t{u}-\h{u})-p+q=0,\\
\text{H2:} \ \ & (u-\h{\t{u}})(\t{u}-\h{u})-(p-q)(u+\t{u}+\h{u}+\h{\t{u}}+p+q)=0,\\
\text{H3($\delta$):} \ \ & p(u\t{u}+\h{u}\h{\t{u}})-q(u\h{u}+\t{u}\h{\t{u}})+\delta(p^2-q^2)=0,\\
\text{A1($\delta$):} \ \ & p(u+\h{u})(\t{u}+\h{\t{u}})-q(u+\t{u})(\h{u}+\h{\t{u}})-\delta^2pq(p-q)=0, \\
\text{A2:} \ \ & p(1-q^2)(u\t{u}+\h{u}\h{\t{u}})-q(1-p^2)(u\h{u}+\t{u}\h{\t{u}}) \nonumber \\
&-(p^2-q^2)(1+u\t{u}\h{u}\h{\t{u}})=0,\\
\text{Q1($\delta$):} \ \ & p(u-\h{u})(\t{u}-\h{\t{u}})-q(u-\t{u})(\h{u}-\h{\t{u}})+\delta^2pq(p-q)=0, \label{Q1}\\
\text{Q2:} \ \ & p(u-\h{u})(\t{u}-\h{\t{u}})-q(u-\t{u})(\h{u}-\h{\t{u}}) \nonumber \\
&+pq(p-q)(u+\t{u}+\h{u}+\h{\t{u}}-p^2+pq-q^2)=0,\\
\text{Q3($\delta$):} \ \ & p(1-q^2)(u\h{u}+\t{u}\h{\t{u}})-q(1-p^2)(u\t{u}+\h{u}\h{\t{u}}) \nonumber \\
&-(p^2-q^2)\left(\t{u}\h{u}+u\h{\t{u}}+\frac{\delta^2(1-p^2)(1-q^2)}{4pq}\right)=0,\\
\text{Q4:} \ \ & \mathrm{sn}(p)(u\t{u}+\h{u}\h{\t{u}})-\mathrm{sn}(q)(u\h{u}+\t{u}\h{\t{u}})-\mathrm{sn}(p-q)(\t{u}\h{u}+u\h{\t{u}} ) \nonumber \\
&+\mathrm{sn}(p)\mathrm{sn}(q)\mathrm{sn}(p-q)(1+k^2u\t{u}\h{u}\h{\t{u}})=0,\label{Q4}
\end{align}
\end{subequations}
where $\delta$ is an arbitrary parameter, $\mathrm{sn}(p)=\mathrm{sn}(p;k)$ is the Jacobi elliptic function
and $Q4$ in the above form was given by Hietarinta \cite{Hie05}.
Consistency-around-the-cube is viewed as integrability of the ABS equations.
If the top equation is considered to be the same as the bottom equation but with $\b u$ as a new solution,
then the side equations as a coupled system, i.e.
\begin{equation}
Q(u,\t{u},\b{u},\t{\b{u}};p,r)=0,~~ Q(u,\h{u},\b{u},\h{\b{u}};q,r)=0
\end{equation}
automatically provides a B\"{a}cklund transformation (BT) for the bottom equation \eqref{Q}.
Its linearized form obtained by introducing $\b u=g/f$ acts as a Lax pair of equation \eqref{Q}.
If imposing different ABS equations on the bottom and top  respectively, and connecting them via suitable limit procedures,
then the two side equations will give a BT of the two original equations.
This beautiful idea was demonstrated in \cite{A08} by Atkinson.
Usually by an auto-BT we mean it connects different solutions of same equation while by a nonauto-BT we mean it connects solutions of two different equations.

It is well known that BT originated from the construction  of pseudo-spherical surfaces
and BTs have been playing important roles in  soliton theory \cite{Miu-1976,PruS-1998,RogS-2002}.
In this paper, we will consider the ABS list and focus on those BTs that can be constructed by using
decomposition property of the ABS equations.
In fact, each equation $Q=0$ in the ABS list admits a decomposition which is an analogue of the following,
\begin{equation}
\mathcal{H}= h^{12}h^{34}-h^{13}h^{24}=PQ, \ \ P=\left|\begin{array}{ccc}
   Q& Q_{u_1}&Q_{u_4}\\
   Q_{u_2}& Q_{u_1u_2}&Q_{u_2u_4}\\
   Q_{u_3}& Q_{u_1u_3}&Q_{u_3u_4}\\
  \end{array}
\right|, \label{PQ}
\end{equation}
where $u_1=u,~ u_2=\t u,~ u_3=\h u,~ u_4=\th u$,
\begin{equation}
h^{ij}(u_i,u_j)=Q_{u_k}Q_{u_l}-QQ_{u_ku_l},
\label{hQ}
\end{equation}
$Q_{u_k}=\partial_{u_k}Q,~ Q_{u_iu_k}=\partial_{u_i}\partial_{u_k}Q$, and $i,j,k,l$ are distinct elements in $\{1,2,3,4\}$.
The composition \eqref{PQ} played a crucial role in the classification of the ABS list and further
discussions with less restriction than the CAC condition \cite{ABS09,RB11,Boll-phd}.
$h^{ij}$ are used to study singularity structures of solutions,
which is considered to be related to the CAC property and boundary value structures (cf.\cite{A11,AJ13}).
The decomposition \eqref{PQ} and $h$ polynomials can be used to construct BTs.
For example, H3($\delta$) has the following decomposition:
\[\mathcal{H}=h(u,\t{u},p) {h}(\h u,\th{u},p)-h(u,\h{u},q) {h}(\t u,\th{u},q) = H3(\delta)=0,\]
where $h$ is defined by $h(u,\t{u},p)=u\t{u}+p\delta$.
This implies that with unknown function $U$, the pair
\begin{equation}
h(u,\t{u},p)=  U\t U,~~ h(u,\h{u},q)= U \h U
\label{h-U}
\end{equation}
will provide  a BT to connect H3($\delta$) and the $U$-equation
which comes from the compatibility of  the pair in term of $u$.
In this paper we will examine the system \eqref{h-U} for  some ABS equations.
As a result, when $h$ is affine linear, a completed list of BTs together with CAC equations that they connect are given.
When $h$ is beyond  affine linear, some quadratic CAC equations are obtained.
Some BTs can be used to construct new solutions  and weak Lax pairs.

The paper is organized as follows.
In Sec.\ref{sec-2} we revisit decomposition of the ABS list.
In Sec.\ref{sec-3} we discuss possible forms of $h$ and the related quadrilateral equations of $u$ and $U$,
which are listed in Table 1 and 2.
Sec.\ref{sec-4} includes some examples as applications, a new weak Lax pair of Q1(0), new polynomial solutions
of  Q1($\delta$) and rational solutions of H3*($\delta$) in Casoratian form.
Finally, Sec.\ref{sec-5} is for conclusions.

\section{Decomposition \eqref{PQ} of the ABS list}\label{sec-2}

\subsection{Decomposition \eqref{PQ}: revisited}\label{sec-2-1}

Let us revisit the decomposition \eqref{PQ} and have a look at the relation of $Q$ and  $P$ from a generic viewpoint.
Consider  $Q(u,\t u,\h u, \th u; p,q)$ to be a general quadrilateral affine linear polynomial:
\begin{align}
Q(u,\t u,\h u, \th u; p,q)=& k u\t u \h u\th u+l_1u\t u\h u+l_2u\t u\th u+l_3u\h u\th u+l_4\t u\h u\th u \nonumber\\
& +p_1 u \t u + p_2 \t u \h u + p_3 \h u \th u + p_4 u \th u + p_5 u \h u + p_6 \t u\th u \nonumber\\
& +q_1 u +q_2 \t u+ q_3 \h u+ q_4 \th u +c,
\label{Q-gen}
\end{align}
where $k, l_i, p_i, q_i, c$ are constants.
Let $\mathcal{P}_s^t$ denote a set of polynomials with $s$ distinct variables in $\{u,\t u,\h u, \th u\}$
and at most degree $t$ for each variable.
With this definition, the most general element in $\mathcal{P}_4^1$ is $Q$ defined in \eqref{Q-gen} and  $h^{ij}$ belongs to $\mathcal{P}_2^2$.
For the above general $Q$, the decomposition \eqref{PQ} holds and $P\in \mathcal{P}_4^1$ \cite{ABS09}.
In fact, the discriminant of $h^{ij}$ plays an important role in the classification of integrable quadrilateral equations.
With less restriction than the CAC condition, the classification of $Q$-type equations was done in  \cite{ABS09}
and a full classification was finished in \cite{RB11,Boll-phd}.

In the following let us take a close look at the relation between $Q$ and $P$ in \eqref{PQ}. Similar to \eqref{hQ} we define
\begin{equation}
g^{ij}(u_i,u_j)=P_{u_k}P_{u_l}-P P_{u_ku_l},
\label{gP}
\end{equation}
where $P$ is defined in \eqref{PQ}. Then we have the following.

\begin{proposition}\label{ProPQ1}
For the polynomial $Q$ given in \eqref{Q-gen}, $P$ defined in \eqref{PQ} and $h^{ij}$ in \eqref{hQ}, there is a constant $K$ such that
\[g^{ij}=(-1)^{j-i}K h^{ij},~~ (i,j)\in \{(1,2), (1,3), (2,4), (3,4)\}.\]
Particularly, when $K=0$, $P$ can be factorized as a product of distinct linear function $a_iu_i+b_i$.
\end{proposition}

\begin{proof}
It has been proved that $P\in \mathcal{P}_4^1$  \cite{ABS09}.
Thus we can switch the roles of $Q$ and $P$ in \eqref{PQ}.
For the polynomial $Q$ given in \eqref{Q-gen}, $P$ defined in \eqref{PQ} and $h^{ij}$ in \eqref{hQ},
by  direct calculation we find
\begin{equation}
    g^{ij}=(-1)^{j-i}K h^{ij},~~ (i,j)\in \{(1,2), (1,3), (2,4), (3,4)\},
    \label{gh}
\end{equation}
with a same  $K$ which is an irreducible rational function of $\{u_i\}$.
Similarly we find that
This implies
\begin{equation*}
  g^{12}g^{34}-g^{13}g^{24}=K^2QP.
\end{equation*}
Corresponding to the structure in \eqref{PQ}, there must be
\begin{equation*}
K^2Q=\left|\begin{array}{ccc}
   P& P_{u_1}&P_{u_4}\\
   P_{u_2}& P_{u_1u_2}&P_{u_2u_4}\\
   P_{u_3}& P_{u_1u_3}&P_{u_3u_4}\\
  \end{array}
  \right| \in\mathcal{P}_4^1.
\end{equation*}
Since $Q\in\mathcal{P}_4^1$ the only choice for $K$ is a constant.

If $K=0$, then $g^{ij}=0$ in light of \eqref{gh}.
Noticing that
\begin{equation*}
   ( \ln P)_{u_ku_l}
   =\frac{g^{ij}}{P^2},
\end{equation*}
if  $g^{ij}=0$ we have
\begin{equation*}
\ln P = \psi_1(u_1)+\psi_2(u_2)+\psi_3(u_3)+\psi_4(u_4),
\end{equation*}
where $\psi_i(u_i)$ is a function of $u_i$.
This means  $P\in \mathcal{P}_4^1$ can be factorized as
\begin{equation*}
P=\phi_1(u_1)\phi_2(u_2)\phi_3(u_3)\phi_4(u_4),
\end{equation*}
where there must be $\phi_i(u_i)=a_iu_i+b_i$ because $P\in \mathcal{P}_4^1$.

\end{proof}

The above proposition reveals an ``adjoint" relation between $Q$ and $P$ if $Q$ is an affine-linear quadrilateral polynomial \eqref{Q-gen}.
For A-type and Q-type  ABS equations, one can see that  $Q$ and $P$ are almost same.

\subsection{Decomposition of the ABS list}\label{sec-2-2}

For each equation \eqref{Q} in the ABS list  it holds that \cite{ABS03}
\begin{equation}\label{hu}
   \mathcal{H}=h(u,\t{u},p) {h}(\h u,\th{u},p)-h(u,\h{u},q) {h}(\t u,\th{u},q)=0,
\end{equation}
where function $h(u,\tilde{u},p)$ is $h^{12}$ divided by certain factor $\kappa(p,q)$.
Functions $h$, $\mathcal{H}$ and relations with  the ABS equations are listed below:
\begin{subequations}\label{decom}
\begin{align}
\text{H1}:\ \ &h(u,\t{u},p)=1, \mathcal{H}=0,\ \  \text{identity}, \\
\text{H2}:\ \ &h(u,\t{u},p)=u+\t{u}+p, \ \ \mathcal{H}= H2(u,\t{u},\h{u},\h{\t{u}};p,q)=0,\\
\text{H3}:\ \ &h(u,\t{u},p)=u\t{u}+p\delta, \ \ \mathcal{H}= H3(u,\t{u},\h{u},\h{\t{u}};p,q)=0,\\
\text{A1}:\ \ &h(u,\t{u},p)=\frac{1}{p}(\t{u}+u)^2-\delta^2p, \ \ \mathcal{H}=A1(u,\t{u},\h{u},\h{\t{u}};p,q)A1(u,\t{u},\h{u},\h{\t{u}};p,-q)=0,\\
\text{A2}:\ \ &h(u,\t{u},p)=\frac{(pu\t{u}-1)(u\t{u}-p)}{1-p^2}, \ \ \mathcal{H}=A2(u,\t{u},\h{u},\h{\t{u}};p,q)A2(u,\t{u},\h{u},\h{\t{u}};p,q^{-1})=0,\\
\text{Q1}:\ \ &h(u,\t{u},p)=\frac{1}{p}(\t{u}-u)^2-\delta^2p, \ \ \mathcal{H}=Q1(u,\t{u},\h{u},\h{\t{u}};p,q)Q1(u,\t{u},\h{u},\h{\t{u}};p,-q)=0,
\label{Q1H}\\
\text{Q2}:\ \ &h(u,\t{u},p)=\frac{1}{p}(\t{u}-u)^2-2p(u+\t{u})+p^3, \nonumber\\
& \mathcal{H}=Q2(u,\t{u},\h{u},\h{\t{u}};p,q)Q2(u,\t{u},\h{u},\h{\t{u}};p,-q)=0,\\
\text{Q3}:\ \ &h(u,\t{u},p)=\frac{p}{1-p^2}(u^2+\t{u}^2)-\frac{1+p^2}{1-p^2}u\t{u}+\frac{(1-p^2)\delta^2}{4p},\nonumber\\
& \mathcal{H}=Q3(u,\t{u},\h{u},\h{\t{u}};p,q)Q3(u,\t{u},\h{u},\h{\t{u}};p,q^{-1})=0,\\
\text{Q4}:\ \ &h(u,\t{u},p)=\frac{-1}{\mathrm{sn}(p)}\left(k^2\mathrm{sn}^2(p)u^2\t{u}^2+2\mathrm{sn}'(p)u\t{u}-u^2-\t{u}^2+\mathrm{sn}^2(p)\right),\nonumber\\
& \mathcal{H}=Q4(u,\t{u},\h{u},\h{\t{u}};p,q)Q4(u,\t{u},\h{u},\h{\t{u}};p,-q)=0.
\end{align}
\end{subequations}

We note that for the ABS equations, the case $K=0$ in Proposition \ref{ProPQ1}
corresponds to H-type equations in the ABS list;
for H-type equations $P$ is a constant and for H1 even $P=0$;
for A-type and Q-type equations, $Q$ and $P$ differ only in the  parameter $q$.

\section{B\"acklund transformations}\label{sec-3}

Motivated by the decomposition \eqref{decom} of the ABS equations, we consider the following system
\begin{subequations}\label{BTuU}
\begin{align}
& h(u,\t{u},p)=U\t{U},\\
& h(u,\h{u},q)=U\h{U},
\end{align}
\end{subequations}
where to meet the consistency w.r.t. $u$ we request $U$ satisfies certain quadrilateral equation
\begin{equation}
F(U,\t U,\h U, \th U; p,q)=0,
\label{U-eq}
\end{equation}
which we call $U$-equation for convenience.
In fact, on one hand, for arbitrary $U$ the function $h$ defined by \eqref{BTuU} satisfies  \eqref{hu}.
On the other hand, since $u$ must be well defined by \eqref{BTuU},
the two equations in \eqref{BTuU} must be compatible (i.e. $\th u=\t{\h u}$),
which leads to the $U$-equation.
In this paper by the \textit{consistent triplet} generated by \eqref{BTuU}
we mean the triplet composed by $u$-equation \eqref{Q}, $U$-equation \eqref{U-eq} and their BT \eqref{BTuU},
in which the compatibility of \eqref{BTuU} w.r.t. $u$ yields \eqref{U-eq}
and the compatibility of \eqref{BTuU} w.r.t. $U$ yields \eqref{Q} (cf.\cite{ZhaZ-arxiv-17}).
Obviously, not any BT can generate a consistent triplet.

As for generating solutions, we note that for H2 and H3, $u$  solved from \eqref{BTuU} with corresponding $h$ will provide a solution to
these two equations,
while for the rest equations in the ABS list there is uncertainty.
For example, for Q1, we do not know whether $u$ solves Q1 or $\mathrm{Q1}(u,\t{u},\h{u},\h{\t{u}};p,-q)=0$.
In this section, instead of finding solutions, we are more interested in considering \eqref{BTuU}
 as a BT to connect $u$-equation \eqref{hu} and $U$-equation.
In the following subsections Sec.\ref{sec-3-1},\ref{sec-3-2},
we start from  a generic  affine-linear polynimial
\begin{equation}
h(u,\t{u},p)=s_0(p)+s_1(p)u+s_2(p)\t{u}+s_3(p)u\t{u},
\label{h-eq}
\end{equation}
where $s_i(p)$ are functions of $p$,
and examine all possibility that admits a consistent triplet in which both  $u$-equation and $U$-equation are CAC.
We note that $h=0$ with \eqref{h-eq} is a discrete Riccati equation,
and the special case $s_1=s_2$ with $s_3=0$ was already considered in \cite{A08}.

\subsection{Consistent triplets}\label{sec-3-1}

When $h$ is defined as \eqref{h-eq}, for the relation of $h$ and possible forms of $u$-equation and  $U$-equation, we have the following.

\begin{theorem}\label{thm1}
When $h(u,\t{u},p)$ in system \eqref{BTuU} is defined by \eqref{h-eq},
then $U$-equation \eqref{U-eq} is affine-linear if and only if
either
\begin{equation}
h(u,\t{u},p)=s_0(p)+s_1(p)u+s_2(p)\t{u}
\label{h1}
\end{equation}
or (after a constant shift $u\to u-c$)
\begin{equation}
h(u,\tilde{u},p)=s_0(p)+s_3(p)u\t{u}.
\label{h2}
\end{equation}
\end{theorem}
\begin{proof}
When $h(u,\t{u},p)$  is defined by \eqref{h-eq},
we  solve out from \eqref{BTuU} that
\begin{equation*}
   \t{u}=\frac{U\t{U}-s_0(p)-s_1(p)u}{s_2(p)+s_3(p)u},\ \ \h{u}=\frac{U\h{U}-s_0(q)-s_1(q)u}{s_2(q)+s_3(q)u}.
\end{equation*}
From  the consistency $\h{\t{u}}=\t{\h{u}}$, we have the resulting equation for $U$: $f(U,\t{U},\h{U}, \h{\t{U}},u)=0$,
which should be  independent of $u$. So
the coefficient of $u$ in $f$ should be zero, which leads to
\begin{equation}
    s_2(q)s_3(p)-s_3(q)s_1(p)=s_2(p)s_3(q)-s_3(p)s_1(q)=0.
    \label{s-eq}
\end{equation}
When $s_3=0$, \eqref{h-eq} turns to be \eqref{h1}.
In this case, the system \eqref{BTuU} is a BT for $u$-equation
\begin{equation}\label{u1}
    \frac{s_0(p)+s_1(p)u+s_2(p)\t{u}}{ s_0(q)+s_1(q)u+s_2(q)\h{u}}=\frac{s_0(q)+s_1(q)\t{u}+s_2(q)\h{\t{u}}}{s_0(p)+s_1(p)\h{u}+s_2(p)\h{\t{u}}}
\end{equation}
and $U$-equation
\begin{equation}\label{U1}
  U(s_1(q)\t{U}-s_1(p)\h{U})+\h{\t{U}}(s_2(q)\h{U}-s_2(p)\t{U})+s_0(q)(s_1(p)+s_2(p))-s_0(p)(s_1(q)+s_2(q))=0.
\end{equation}
Here we note that equation \eqref{u1}, \eqref{U1} and the BT \eqref{BTuU}
compose a consistent triplet (cf.\cite{ZhaZ-arxiv-17}),
i.e. viewing the BT \eqref{BTuU} as a two-component system, then the compatibility of each component yields a
lattice equation of the other component which is in the triplet.

When function $s_3\neq 0$, we have
\begin{equation*}
   c=\frac{ s_2(q)}{s_3(q)}=\frac{ s_1(p)}{s_3(p)}=\frac{ s_1(q)}{s_3(q)}=\frac{ s_2(p)}{s_3(p)},
\end{equation*}
where $c$ is a constant independent of $p$ and $q$. It then follows that $s_2=s_1=cs_3$.
Thus \eqref{h-eq} yields 
$$h(u,\t{u},p)=s_0(p)-c^2s_3(p)+s_3(p)(u+c)(\t{u}+c),$$
which then reduces to \eqref{h2} by redefine $s_0(p) \to s_0(p)+c^2s_3(p)$ and $u\to u-c$.
Consequently the $u$-equation reads
\begin{equation}\label{u2}
    (s_0(p)+s_3(p)u\t{u}) (s_0(p)+s_3(p)\t{u}\h{\t{u}})=  (s_0(q)+s_3(q)u\h{u}) (s_0(q)+s_3(q)\t{u}\h{\t{u}})
\end{equation}
and $U$-equation reads
\begin{align}
&(s_3^2(p)-s_3^2(q))U\t{U}\h{U}\h{\t{U}}+U(s_0(p)s_3^2(q)\t{U}-s_0(q)s_3^2(p)\h{U})\nonumber\\
& +\h{\t{U}}((s_0(p)s_3^2(q)\h{U}-s_0(q)s_3^2(p)\t{U}))=s_0^2(p)s_3^2(q)-s_0^2(q)s_3^2(p).\label{U2}
\end{align}
Note that replacing  $s_0$ with $-\frac{s_0}{s_3}$ and $s_3$ with $\frac{1}{s_3}$,
equation  \eqref{u2} becomes \eqref{U2}.
Eqs.\eqref{u2}, \eqref{U2} and \eqref{BTuU} compose a consistent triplet as well.

\end{proof}

Now we have obtained four quadrilateral equations, \eqref{u1}, \eqref{U1}, \eqref{u2} and \eqref{U2},
all of which are derived as a compatibility of \eqref{BTuU}.
Among them, equation \eqref{u1} with $s_0=1, ~ s_1(p)=p-a,~ s_2(p)=p+a$ can be considered as the
Nijhoff-Quispel-Capel (NQC) equation with $b=a$ (cf. \cite{NQC-1983} and eq.(9.49) in \cite{HieJN-2016}).

\subsection{Multidimensional  consistency with  $h$  given in \eqref{h1} and \eqref{h2}}\label{sec-3-2}

Although  \eqref{u1}, \eqref{U1}, \eqref{u2} and \eqref{U2}
are derived as a compatibility of \eqref{BTuU},
it is not true that they are multidimensionally consistent for arbitrary $s_i$.
After a case-by-case investigation of the CAC property of the four equations,
we reach a full list that includes all multidimensionally  consistent equations when  $h$ are given in \eqref{h1} and \eqref{h2},
which is presented in the  following theorem:

\begin{theorem}\label{class}
For the system \eqref{BTuU} where $h$ is affine-linear as given in \eqref{h1} and \eqref{h2},
if it generates a consistent triplet and acts as a BT between quadrilateral equations which are CAC,
the exhausted results are
{\small
\begin{center}
\setlength{\tabcolsep}{12pt}
\renewcommand{\arraystretch}{2.0}
\begin{tabular}{llll}
\hline
   No.& \qquad   BT\eqref{BTuU}      &        $u$-equation  &   $U$-equation \\
\hline
1&\!\!\!\!$\frac{1}{p}(u-\t{u})=U\t{U}$&Q1($0;p^2,q^2$)& lpmKdV\\
2&\!\!\!\!$u+\t{u}+p=U\t{U}$& H2 & H1($2p,2q$)\\
3&\!\!\!\!$\frac{1}{p}(u+\t{u}-\delta p)=U\t{U}$& A1($\delta;p^2,q^2$) & H3($\delta;2p,2q$)\\
4&\!\!\!\!$u\t{u}+\delta p=U\t{U}$& H3($\delta$) & H3($-\delta$)\\
5&\!\!\!\!$\frac{1}{p}(u\t{u}-1)=U\t{U}$& \eqref{new} &  H3($1$) with $U\to U^{-1}$ \\
6&\!\!\!\!$\frac{1}{\sqrt{1-p^2}}(1-pu\t{u})=U\t{U}$  & A2 &  A2($\sqrt{1-p^2}, \sqrt{1-q^2}$)\\
\hline
\end{tabular}
\begin{center}{\scriptsize {\bf Table 1.}   Consistent triplets }
\end{center}
\end{center}
}
\end{theorem}

Proof of the theorem and equation \eqref{new} are given in Appendix \ref{sec-A-1}

\subsection{Other cases: Q1$(\delta)$, A1$(\delta)$ and A2}

For equations Q1$(\delta)$, A1$(\delta)$ and A2, their $h$ polynomials are not affine linear.
We discuss them one by one.

First, for Q1$(\delta)$, the corresponding system \eqref{BTuU} is
\begin{equation}\label{BTQ1}
(\t{u}-u)^2-\delta^2p^2=pU\t{U},~~
(\h{u}-u)^2-\delta^2q^2=qU\h{U},
\end{equation}
which is quadratic w.r.t. $u$.
We find if
$u$ satisfies Q1$(\delta)$, then $U$ satisfies
\begin{equation}\label{Q1U}
\left[p(U\t{U}-\h{U}\h{\t{U}})-q(U\h{U}-\t{U}\h{\t{U}})\right]^2
+4pq(\t{U}-\h{U})\left[\delta^2(p-q)(U-\h{\t{U}})-U\h{\t{U}}(\t{U}-\h{U})\right]=0,
\end{equation}
and vice versa.
The above equation  can be transformed to H3*$(\delta)$ equation
\begin{equation}\label{H3*}
  (p-q)\left[p(U\h{U}-\t{U}\h{\t{U}})^2-q(U\t{U}-\h{U}\h{\t{U}})^2\right]
+(U-\h{\t{U}})(\t{U}-\h{U})\left[(U-\h{\t{U}})(\t{U}-\h{U})pq-4\delta^2(p-q)\right]=0,
\end{equation}
by transformation $ p\to 1/p, q\to 1/q$.
Here we note that H3*$(\delta)$ is one of integrable quad equations that are multi-quadratic  counterparts
of the ABS equations.
These multi-quadratic equations are consistent in multi-dimensions as well and were systematically found in the work  \cite{AtkN-IMRN-2014}.
\eqref{BTQ1} provides a BT between Q1$(\delta)$ and H3*$(\delta)$ \eqref{Q1U}.

Similarly, for A1, we can find
\begin{equation}\label{BTA1}
(\t{u}+u)^2-\delta^2p^2=pU\t{U},~~
(\h{u}+u)^2-\delta^2q^2=qU\h{U},
\end{equation}
provides a BT between A1 and H3*$(\delta)$ \eqref{Q1U}.

For A2, in the system \eqref{BTuU} there is
\begin{equation}
h(u,\t{u},p)=\frac{(pu\t{u}-1)(u\t{u}-p)}{1-p^2}.
\label{h-A2}
\end{equation}
It is hard to write out a $U$-equation in a neat form.
However, observing that in \eqref{BTuU} $U$ is arbitrary,
we can replace $U$ with $U/f(u)$ where $f(u)$ is a suitable function of $u$
so that the deformed BT
\begin{equation}
h(u,\t{u},p)f(u)f(\t{u})=U\t U,~~h(u,\h{u},q)f(u)f(\h{u})=U\h U
\end{equation}
yields a $U$-equation with a neat form.
Taking $f(u)=1/u$, \eqref{BTuU} with \eqref{h-A2} becomes
\begin{equation}
\frac{(pu\t{u}-1)(u\t{u}-p)}{(1-p^2)u\t{u}}=U\t{U},~~\frac{(qu\h{u}-1)(u\h{u}-q)}{(1-q^2)u\h{u}}=U\t{U},
\label{bt-1}
\end{equation}
connects the solutions between A2($u$)   and A2*($U$) equation \cite{AtkN-IMRN-2014}
\begin{align}\label{H3*}
& (p-q)\left[p(U\h{U}-\t{U}\h{\t{U}})^2-q(U\t{U}-\h{U}\h{\t{U}})^2\right] \nonumber\\
&~~~~+(U-\h{\t{U}})(\t{U}-\h{U}) \left[(U-\h{\t{U}})(\t{U}-\h{U})(pq-1)+2(p-q)(1+U\t{U}\h{U}\h{\t{U}})\right]=0,
\end{align}
with  $p\to 2(p^2+1)/(1-p^2) , q\to 2(q^2+1)/(1-q^2)$.
One more example of utilizing $f(u)$ is H3($\delta$). For the BT of No.4 in Table 1,
taking $f(u)=u$ we have
\begin{equation}
(u\t{u}+\delta p)u\t{u}=U\t{U},~~(u\h{u}+\delta q)u\h{u}=U\h{U},
\label{bt-2}
\end{equation}
connects the solutions between H3$(\delta)$($u$) and H3*$(\delta)$($U$) equation with parameters $p\to 4/p^2, q\to 4/q^2$.
BTs \eqref{bt-1} and \eqref{bt-2} have been found in \cite{AtkN-IMRN-2014}.
Let us look at a third example which is not known before.
It can be verified that
\begin{equation}
\frac{u\t{u}-p}{pu\t{u}-1}=U\t{U},~~ \frac{u\t{u}-q}{qu\t{u}-1}=U\h{U}
\label{bt-a2}
\end{equation}
is an auto-BT of A2. Taking  {$f(u)=1/u$ } the new BT provides a transformation
between A2$(u)$  and  A2*($U$)  with parameters $p\to 2p^2-1, q\to 2q^2-1$.

We collect the BTs of this subsection in Table 2.

{\small\begin{center}
\begin{center}
\setlength{\tabcolsep}{12pt}
\renewcommand{\arraystretch}{2.0}
\begin{tabular}{llll}
\hline
   No.& \qquad   BT\eqref{BTuU}      &        $u$-equation  &   $U$-equation \\
\hline
1&$\frac{1}{p}(\t{u}-u)^2-\delta^2p=U\t{U}$&Q1($\delta$)& H3*($\delta;\frac{1}{p},\frac{1}{q}$)\\
2&$\frac{1}{p}(\t{u}+u)^2-\delta^2p=U\t{U}$&A1($\delta$) & H3*($\delta;\frac{1}{p},\frac{1}{q}$)\\
3&$\frac{(pu\t{u}-1)(u\t{u}-p)}{(1-p^2)u\t{u}}=U\t{U}$& A2 & A2*$\left(\frac{2(p^2+1)}{1-p^2},\frac{2(q^2+1)}{1-q^2}\right)$\\
4&$(u\t{u}+\delta p)u\t{u}=U\t{U}$& H3($\delta$) & H3*$(\delta;\frac{4}{p^2},\frac{4}{q^2})$\\
5&$\frac{u\t{u}-p}{pu\t{u}-1}=U\t{U}$& A2 & A2\\
6&$\frac{u\t{u}-p}{(pu\t{u}-1)u\t{u}}=U\t{U}$& A2 &  A2*$(2p^2-1,2q^2-1)$\\
\hline
\end{tabular}
\begin{center}{\scriptsize {\bf Table 2.}   BT\eqref{BTuU}  related to Q1$(\delta)$, A1$(\delta)$, H3($\delta$) and A2.}
\end{center}
\end{center}
\end{center}
}

In this section we have given an exhausted examination for the case where  $h$ is the affine-linear polynomial \eqref{h-eq}. 
For Q1($\delta$), A1($\delta$) and A2,  their $h$ polynomials are not affine linear
and their corresponding $U$-equations are usually multi-quadratic  counterparts of the ABS equations, (see Table 2).
For Q2, Q3 and Q4, their $h$ polynomials are so complicated that from system  \eqref{BTuU} we can not derive  explicit $U$-equations.

\section{Applications}\label{sec-4}

Consistent triplets have been used as a main tool to find rational solutions for quadrilateral equations (see \cite{ZhaZ-arxiv-17}).
In this section we would like to introduce more applications, which are BT and weak Lax pair of Q1(0),
polynomial solutions of Q1($\delta$) and rational solutions of H3$^*(\delta)$.

\subsection{BT and weak Lax pair of Q1(0)}\label{sec-4-1}

From the previous discussion, we know that Q1(0),
\begin{equation}\label{Q10}
    p^2(u-\h{u})(\t{u}-\h{\t{u}})-q^2(u-\t{u})(\h{u}-\h{\t{u}})=0,
\end{equation}
has a  BT
\begin{equation}
\t{u}-u=pU\t{U},~~ \h{u}-u=qU\h{U},
\label{bt-q1-1}
\end{equation}
where $U$ satisfies the lpmKdV equation.
If $U$ solves the lpmKdV equation, so does $1/U$.
Employing this symmetry we introduce
\begin{equation}
    \t{\b{u}}-\b{u}= pU^{-1}\t{U}^{-1},\quad  \h{\b{u}}-\b{u}=  qU^{-1}\h{U}^{-1}
\label{bt-q1-2}
\end{equation}
as an adjoint system of \eqref{bt-q1-1}, which is also a BT between Q1(0) and lpmKdV.
Eliminating $U$ we reach
\begin{equation*}
    (\t{\b{u}}-\b{u})(\t{u}-u)= p^2,\quad   (\h{\b{u}}-\b{u})(\h{u}-u)=q^2,
\end{equation*}
which is an auto-BT of Q1(0).
Noticing the symmetry that $u$ and $1/u$ can solve Q1(0) simultaneously,
we replace $\bar{u}$ with ${1}/{\bar{u}}$ and reach
\begin{equation}\label{auto-BTQ10}
    (\t{\b{u}}-\b{u})(\t{u}-u)+ p^2\b{u}\,\t{\b{u}}=0,\quad
    (\h{\b{u}}-\b{u})(\h{u}-u)+ q^2\b{u}\,\h{\b{u}}=0,
\end{equation}
which is another auto-BT of Q1(0).
One can check that the following 6 equations
\begin{subequations}\label{Q10cube}
\begin{align}
& p^2(u-\h{u})(\t{u}-\h{\t{u}})-q^2(u-\t{u})(\h{u}-\h{\t{u}})=0,\label{Q10cube-1}\\
& (\t{\b{u}}-\b{u})(\t{u}-u)+ p^2\b{u}\,\t{\b{u}}=0,\label{Q10cube-2}\\
& (\h{\b{u}}-\b{u})(\h{u}-u)+ q^2\b{u}\,\h{\b{u}}=0,\label{Q10cube-3}\\
& (\h{\t{\b{u}}}-\h{\b{u}})(\h{\t{u}}-\h{u})+ p^2\h{\b{u}}\,\h{\t{\b{u}}}=0,\label{Q10cube-4}\\
& (\h{\t{\b{u}}}-\t{\b{u}})(\h{\t{u}}-\t{u})+ q^2\t{\b{u}}\,\h{\t{\b{u}}}=0,\label{Q10cube-5}\\
& p^2(\b{u}-\h{\b{u}})(\t{\b{u}}-\h{\t{\b{u}}})-q^2(\b{u}-\t{\b{u}})(\h{\b{u}}-\h{\t{\b{u}}})=0 \label{Q10cube-6}
\end{align}
\end{subequations}
can be consistently embedded on 6 faces of a cube.

The BT \eqref{bt-q1-2} yields a pair of linear problems (Lax pair):
\begin{equation}\label{cube-laxpair}
 \t{\Phi}=\left(\begin{array}{cc}
   1& 0\\
 \frac{p^2}{\t{u}-u}& 1\\
   \end{array}\right)\Phi,~~
 \h{\Phi}=\left(\begin{array}{cc}
   1& 0\\
 \frac{q^2}{\h{u}-u}& 1\\
   \end{array}\right)\Phi,
\end{equation}
where $\Phi=(g, f)^T$. The consistency of \eqref{cube-laxpair}
leads to an equation
\begin{equation}\label{condition}
    (u-\t{u}-\h{u}+\h{\t{u}})\left[p^2(u-\h{u})(\t{u}-\h{\t{u}})-q^2(u-\t{u})(\h{u}-\h{\t{u}})\right]=0,
\end{equation}
which is Q1(0) multiplied by a factor $(u-\t{u}-\h{u}+\h{\t{u}})$.
This means Q1(0) can not be fully determined by \eqref{cube-laxpair}.
Such an Lax pair is called a weak Lax pair and was first systematically studied in \cite{weak-lax}.
\eqref{cube-laxpair} is a new weak Lax pair of Q1(0).
As a result,
 replacing \eqref{Q10cube-1} and \eqref{Q10cube-6} by
\[u-\t{u}-\h{u}+\h{\t{u}}=0,~~ \b{u}\,\t{\b{u}}\,\h{\b{u}}\,\h{\t{\b{u}}}(1/\b{u}-1/\t{\b{u}}-1/\h{\b{u}}+1/\h{\t{\b{u}}})=0,\]
respectively,
system \eqref{Q10cube} is also  a consistent cube.

In addition to the weak Lax pair of Q1(0), we  have shown an approach
to construct auto-BT for $u$-equation from \eqref{BTuU} if
$U$-equation admits a symmetry $U\to 1/U$.
For A2 and related \eqref{bt-a2}, employing the same technique, we have relations
\[
(u\t{u}-p)(\b{u}\,\t{\b{u}}-p)=(pu\t{u}-1)(p\b{u}\,\t{\b{u}}-1),~~ (u\h{u}-q)(\b{u}\,\h{\b{u}}-q)=(qu\h{u}-1)(q\b{u}\,\h{\b{u}}-1),\]
and
\[
(u\t{u}-p)(1-p\b{u}\,\t{\b{u}})=(pu\t{u}-1)(p-\b{u}\,\t{\b{u}}),~~ (u\h{u}-q)(1-q\b{u}\,\h{\b{u}})=(qu\h{u}-1)(q-\b{u}\,\h{\b{u}}).\]
Both of them are auto BTs of A2.

\subsection{Polynomial solutions of Q1($\delta$)}

Consider \eqref{BTQ1}, i.e.
\begin{equation}\label{BTQ1-1}
(\t{u}-u)^2-\delta^2p^2=pU\t{U},~~
(\h{u}-u)^2-\delta^2q^2=qU\h{U},
\end{equation}
which is a BT between Q1($\delta$) and H3*($\delta$).
However, if we do not care about what $U$-equation is,
then from decomposition \eqref{Q1H} any $u$ defined by \eqref{BTQ1-1} will be a solution of
\begin{equation}
\mathcal{H}=Q1(u,\t{u},\h{u},\h{\t{u}};p,q)Q1(u,\t{u},\h{u},\h{\t{u}};p,-q)=0.
\label{Q1HH}
\end{equation}
In other words,
\eqref{BTQ1-1} may also be a BT between $Q1(u,\t{u},\h{u},\h{\t{u}};p,-q)=0$ and some $U$-equation other than H3*($\delta$).
This means, if we just solve \eqref{BTQ1-1} and obtain $u$, we should verify
whether $u$ satisfies Q1($\delta$) \eqref{Q1} or $Q1(u,\t{u},\h{u},\h{\t{u}};p,-q)=0$.

To solve \eqref{BTQ1-1} which is a quadratic system , we suppose that $U$ is a polynomial of
\begin{equation}
x=an+bm+\gamma,
\label{x}
\end{equation}
say,
\begin{equation}
U=\displaystyle{\sum_{i=0}^N}c_{N-i}x^i
\label{U}
\end{equation}
with constant $a, b, \gamma, c_i$ and $c_0\neq 0, N\geq0$.
Introduce
\begin{equation}
v_1=\t{u}-u,~~ v_2=\h{u}-u,
\label{vu}
\end{equation}
where  $v_1, v_2$ should satisfy
\begin{equation}\label{v12}
    \h{v}_1-v_1=\t{v}_2-v_2
\end{equation}
due to consistency of \eqref{vu}.
For the case both $v_1$ and $v_2$ are also polynomials of $x$, we have the following result:

\begin{theorem}\label{t-q1}
When $U$ is given in \eqref{U}, we can convert the system \eqref{BTQ1-1} to
\begin{equation}\label{BTQ1-11}
v_1^2-\delta^2p^2=pU\t{U},~~
v_2^2-\delta^2q^2=qU\h{U}.
\end{equation}
When $N\geq 1$ and we require $v_i$ have the following form,
\begin{equation}
v_1=\displaystyle{\sum_{i=0}^N}f_{N-i}x^i,
~v_2= \displaystyle{\sum_{i=0}^N}g_{N-i}x^i
\label{Uvv}
\end{equation}
with constants $f_i, g_i$ to be determined.
Then, the only allowed values for  $N$ are  1 and 2.
$u$ is recovered through \eqref{vu}.
\end{theorem}

The proof for this theorem is given in Appendix \ref{sec-B}.

Let us turn to find polynimial solutions.
When $N$=0, we have $U=c_0$ and
\begin{equation*}
    (\t{u}-u)^2=p(\delta^2p+c_0^2), \ \ (\h{u}-u)^2=q(\delta^2q+c_0^2).
\end{equation*}
Suppose
\begin{equation}
p=\frac{c_0^2}{a^2-\delta^2},~~ q=\frac{c_0^2}{b^2-\delta^2},~~~\alpha=pa,~ \beta=qb.
\end{equation}
It turns out that four possibilities for $u$ are
\begin{align*}
 &\alpha n+\beta m+\gamma, \quad \frac{(-1)^{n+1}}{2}\alpha+\beta m+\gamma,\\
 &\alpha n+\frac{(-1)^{m+1}}{2}\beta+\gamma, \ \ \frac{(-1)^{n+1}}{2}\alpha+\frac{(-1)^{m+1}}{2}\beta+\gamma,
\end{align*}
which coincide with the result in \cite{HZ09}.

When $N=1,2$, with $p,q$ parameterized as in \eqref{pq},  following Theorem \ref{t-q1}, after some calculation and scaling, 
we find solutions to \eqref{BTQ1-1}:
\begin{subequations}\label{H3*-U-sol}
\begin{align}
& u=\pm\delta x^2+\gamma_0, ~~~~ U=\pm2\delta x,\label{H3*-U-sol-1}\\
& u=\frac{c_0}{3}x^3-\frac{\delta^2}{c_0}x-\frac{c_0}{3}(a^3n+b^3m)+\gamma_0, ~~~~ U=c_0x^2-\frac{\delta^2}{c_0}.\label{H3*-U-sol-2}
\end{align}
\end{subequations}
We can check that $u$ and  $U$ respectively
satisfy  Q1($\delta$) and H3*($\delta$) equation \eqref{Q1U}.
These are polynomial solutions.

\subsection{Rational solutions of H3*($\delta$)}\label{sec-4-3}

One can derive rational solutions for H3*($\delta$) from those of Q1($\delta$) and BT  \eqref{BTQ1}.

It has been proved that Q1($\delta$) with $p,q$ parameterized as in \eqref{pq} has the following rational solutions \cite{ZhaZ-arxiv-17}:
\begin{equation}\label{Q1-series-sol}
    u_{N+2}=\frac{~\b{\b f} + \delta^2 \ub{\ub f}~}{f},
\end{equation}
where $N$-th order Casoratian $f$  is given by
\[ f\doteq  f_N =|\h{N-1}|=|\alpha(n,m,0),\alpha(n,m,1),\cdots,\alpha(n,m,N-1)|, \]
for $N\geq1$ and extended to negative  direction by
\begin{equation}
f_{-N}=(-1)^{[\frac{N}{2}]}f_{N-1},~~ f_0=1,
\label{f-neg}
\end{equation}
$[\,\cdot\,]$ denotes the greatest integer function,
$\ub f=f_{N-1}=|\h{N-2}|$, $\b f= f_{N+1}=|\h{N}|$, etc;
the Casoratian vector $\alpha$ is 
\[\alpha(n,m,l)=(\alpha_0, \alpha_1, \cdots, \alpha_{M-1})^T,~~ \alpha_j=\frac{1}{(2j+1)!}\partial^{2j+1}_{s_i}\psi_i|_{s_i=0},~~ M=1,2,\cdots\]
with function $\psi_i$
\begin{equation*}
    \psi_i(n,m,l)=  \psi_i^{+}(n,m,l) + \psi_i^{-}(n,m,l),~~
      \psi_i^{\pm}(n,m,l)=\pm \frac{1}{2}(1\pm s_i)^{l+l_0}(1\pm as_i)^n(1\pm bs_i)^m.
\end{equation*}
The Casoratian $f$ defined above satisfies a superposition relation \cite{ZhaZ-arxiv-17}
\begin{subequations}\label{f-itera}
\begin{align}
& \t{\b{\b{f}}}f-\b{\b{f}}\t{f}=a\t{\b{f}}\,\b{f},\\
& \h{\b{\b{f}}}f-\b{\b{f}}\h{f}=b\h{\b{f}}\,\b{f}.
\end{align}
\end{subequations}

Making use of \eqref{BTQ1}, \eqref{Q1-series-sol}, \eqref{f-neg} and \eqref{f-itera}, by a direct calculation 
we find rational solutions of H3*($\delta$) \eqref{Q1U} can be written as
\begin{equation}\label{H3*-series-sol}
    U_{N+2}=\frac{~\b{ f}^2 - \delta^2 \ub{f}^2~}{f^2},~~ N\in \mathbb{Z}.
\end{equation}
The first three solutions are
\begin{align*}
&U_1=1-x_1^2\delta^2,\\
&  U_2=x_1^2-\delta^2,\\
& U_3=\frac{(x_1^3-x_3)^2-9\delta^2}{9x_1^2},
\end{align*}
where
\[x_i=a^i n+b^i m+\gamma_i.\]
Here $U_2$ is  \eqref{H3*-U-sol-2} with $c_0=1$.

Finally, we note that, compared with the solution of H3($\delta$) given by \cite{ZhaZ-arxiv-17}, which is
\[
     Z_{N+2} = (-1)^{\frac{n+m}{2}+\frac{1}{4}}\frac{~\b f +(-1)^{n+m}\delta \ub{f}~}{f},
\]
when $\delta =\mathrm{i}=\sqrt{-1}$  it is interesting to find the relation
$U_N=|Z_N|^2$.

\section{Conclusions}\label{sec-5}

BTs contain compatibility and are closely related to integrability of the equations that they connect.
In this paper we have investigated system \eqref{BTuU} as a BT.
When $h$ is affine linear with a generic form \eqref{h-eq}, we made a complete examination
and all consistent triplets are listed in Table 1.
As applications, apart from constructing solutions (cf.\cite{ZhaZ-arxiv-17}),
these BTs in the triplets can be viewed as Lax pairs of $u$-equations, where
wave function $\Phi=(g,f)^T$ can be introduced by taking $U=g/f$
but usually it is hard to introduce an significant  spectral parameter.
When $h$ is beyond affine linear, system \eqref{BTuU} as a BT and the connecting quadrilateral equations (including multi-quadratic ones)
are listed in Table 2.
Some BTs are new and were not listed in \cite{A08,AtkN-IMRN-2014}.
Further applications of the obtained BTs, such as constructing weak Lax pair and rational solutions for multi-quadratic lattice equations,
were also shown in the paper.

\vskip 15pt
\subsection*{Acknowledgments}
This project is  supported by the NSF of China (Nos.11371241, 11631007 and  11601312).

\begin{appendix}

\section{Proof of Theorem \ref{class}}\label{sec-A-1}

\subsection{Multidimensional  consistency: $h$  given in \eqref{h1}}\label{sec-3-2A}

The following discussion is on  the basis of the  CAC condition $\b{\h{\t{u}}}=\h{\t{\b{u}}}=\t{\h{\b{u}}}$ for system \eqref{BTCAC}.
First we investigate the case of \eqref{u1} and \eqref{U1} with  $h(u,\t{u},p)$   given in \eqref{h1}
where we assume
$s_1(p)s_2(p)\neq0$, otherwise \eqref{u1} is not a quadrilateral equation.

\subsubsection{$s_0 = 0$}
In this case,  it can be verified that \eqref{u1} always satisfies the CAC condition $\b{\h{\t{u}}}=\h{\t{\b{u}}}=\t{\h{\b{u}}}$.
Canonically, we make a transformation $u\to \left(-\frac{s_1(p)}{s_2(p)}\right)^n\left(-\frac{s_1(q)}{s_2(q)}\right)^mu$
so that  equation \eqref{u1} is in a neat form
\begin{equation*}
    s_1(p)s_2(p) (u-\t{u}) (\h{u}-\h{\t{u}})=s_1(q)s_2(q)(u-\h{u}) (\t{u}-\h{\t{u}}).
\end{equation*}
Without any loss of generality, by assumption of  $s_1(p)=-s_2(p)=\frac{1}{p}$,
equation \eqref{u1} turns to be the equation Q1($0;p^2,q^2$)\footnote{By this we denote Q1(0) in which replacing $p$ and $q$ by
$p^2$ and $q^2$.}, while
the corresponding equation \eqref{U1} becomes the lpmKdV equation,
\begin{equation}\label{lpmKdV}
   p(U\t{U}-\h{U}\h{\t{U}})-q(U\h{U}-\t{U}\h{\t{U}})=0.
\end{equation}

\subsubsection{$s_0\neq 0$}

\textbf{A. $s_1(p)+s_2(p)=ks_0(p)$ with constant  $k$ }

This goes to the case of $s_0=0$ by taking  $u\to u-k^{-1}$ when $k\neq 0$ and
$u\to u-s_0(p)n-s_0(q)m$ when $k=0$.

\vskip 10 pt
\noindent
\textbf{B. $s_1(p)+s_2(p)=ks_0(p)$ with nonconstant $k$}

Check all terms in $ \b{\h{\t{u}}}=\t{\h{\b{u}}}$,  where  the  coefficient of $u\b{u}^3$   reads
\begin{equation}\label{CAC1}
s_1(p)s_2^4(r)A(B+C),
\end{equation}
where
\begin{align*}
A=\,&s_2(p)s_1(p)-s_2(q)s_1(q), ~~ B=(-s_0(p)s_1(p)+s_0(q)s_1(q))(s_2(r)+s_1(r)),\\
C=\,&(s_1^2(p)+s_2(p)s_1(p)-s_1^2(q)-s_2(q)s_1(q))s_0(r)\\
&+\left[s_0(p)s_2(q)+s_1(q)s_0(p)-s_1(p)s_0(q)-s_0(q)s_2(p)\right]s_1(p)s_1(q)s_2^{-1}(r).
\end{align*}
Letting \eqref{CAC1} vanish  leads to only three subcases. \\

\noindent
{\textbf{ Case B.1.~~ $A=0$}}

It directly results in
\begin{equation}
s_2(p)s_1(p)=c_0,~~ \mathrm{with~ constant}~  c_0.
\label{c0}
\end{equation}
Then from  the  coefficient of $\b{u}^3$ we have
\begin{equation}\label{CAC2}
s_1^4(p)s_1^2(q)EF=0,
\end{equation}
where
\begin{align*}
E=\, &s_0(p)s_1(p)(s_1^2(q)+c_0)-s_0(q)s_1(q)(s_1^2(p)+c_0),\\
F=\, &c_0^2s_1(p)s_0(p)-c_0^2s_0(q)s_1(q)-c_0(s_1^2(p)-s_1^2(q))s_0(r)s_1(r)\\
&+(s_0(q)s_1(p)-s_0(p)s_1(q))s_1(p)s_1(q)s_1^2(r).
\end{align*}

If $E=0$, it returns to the Case A. In fact, when $E=0$, under \eqref{c0} we have
\[ \frac{s_1^2(p)+c_0}{s_1(p)s_0(p)}=\frac{s_1^2(q)+c_0}{s_1(q)s_0(q)}=c_1= \frac{s_1(p)+s_2(p)}{s_0(p)}, \]
with constant $c_1$.

In the case that $F=0$ and  $s_1(p)$ is not a constant, it again returns to the Case A.
In fact,  in this case from $F=0$ we can take  $s_0$ to be the form
\[s_0(t)=c_1s_1(t)+c_2s_1^{-1}(t),\]
where $c_1$ and $c_2$ are constants.
Substituting the above with $t=p,q,r$ into  $F=0$ it turns out that $c_2=c_1c_0$,
from which and \eqref{c0} we find $s_1(p)+s_2(p)=s_0(p)/c_1$, which brings the case to Case A.
Thus the only choice is   $s_1(p)$ to be a constant.
Without loss of generality we suppose  $s_1(p)=1$ and as a consequence of \eqref{c0} we also have $s_2(p)=c_0$.
Then, after checking the remaining terms in $\b{\h{\t{u}}}=\h{\t{\b{u}}}=\t{\h{\b{u}}}$
we find $c_0=1$.
Therefore in this case we have $h(u,\t{u},p)=u+\t{u}+p$, and  \eqref{u1} and \eqref{U1} are nothing but H2 and H1($2p,2q$).\\

\noindent
{\textbf{Case B.2. ~~$A\neq 0,~ B=C=0$}}

$B=0$ yields either $s_2(r)+s_1(r)=0$ or $s_0(p)s_1(p)=c_0$ with constant $c_0$.
The former belongs to Case A and then we consider the later, i.e.
\begin{equation}
s_0(p)s_1(p)=c_0.
\label{c0-B}
\end{equation}
Note that if the term $s_1^2(p)+s_2(p)s_1(p)-s_1^2(q)-s_2(q)s_1(q)$ in $C$ vanishes
we will find  $(s_1(p)+s_2(p))s_1(p)$ to be a constant, which, together with \eqref{c0-B}, again leads to Case A.
If the term does not vanish, from $C=0$ we can assume
there are constants $c_2$ and $c_3$ such that $c_2 s_0(r)+ c_3 s^{-1}_2(r)=0$, i.e.
\begin{equation}
s_0(r)s_2(r)=c_1=-c_3/c_2.
\label{c1-B}
\end{equation}
Making use of \eqref{c0-B} and \eqref{c1-B} we reach
\[\frac{c_0(c_1^2-c_0^2)}{c_1}\left(\frac{1}{s_0^2(p)}-\frac{1}{s_0^2(q)}\right)=0.\]
We ignore solution $s_0=c$ because this leads to $s_1$ and $s_2$ to be constants and then brings  the case to Case A.
Therefore we have $c_1=\pm c_0$. Since $c_1=-c_0$ results in $k=0$ which is Case A,
the only choice is $c_1=c_0$ and in this case the canonical form for $h$ can be
$$h(u,\t{u},p)=\frac{1}{p}(\t{u}+u)-\delta p.$$ 
Then it follows that  \eqref{u1} is A1($\delta;p^2,q^2$),
and the corresponding \eqref{U1} is H3($\delta;2p,2q$).\\

\noindent
\textbf{Case B.3.~~$A B\neq 0,~ B+C=0$}

Since $B\neq 0$, from $B+C=0$ we can assume
\begin{equation}
s_1(r)+s_2(r)=c_1s_0(r)+c_2s_2^{-1}(r),
\label{s1s2}
\end{equation}
with constant $c_1$ and nonzero constant $c_2$ (if $c_2=0$ we back to Case A),
in which
\begin{subequations}
\begin{align}
c_1= & \frac{ s_1(p)(s_2(p)+s_1(p))-s_1(q)(s_2(q)+s_1(q))}{s_0(p)s_1(p)-s_0(q)s_1(q)},\label{c1}\\
c_2=& \frac{\left[s_0(p)(s_2(q)+s_1(q))-s_0(q)(s_1(p)+s_2(p))\right]s_1(p)s_1(q)}{s_0(p)s_1(p)-s_0(q)s_1(q)}.\label{c2}
\end{align}
\end{subequations}
Note that in Case B $s_0(r)$ and $s_2^{-1}(r)$ must be linearly independent.
Separate $p$ and $q$ in \eqref{c1} we find
$s_1(p)(s_2(p)+s_1(p))-c_1s_0(p)s_1(p)=c_0$ with nonzero constant $c_0$, which, together with relation \eqref{s1s2},
yields $s_2(p)=c_2c_0^{-1}s_1(p)$. Now, substituting this relation and \eqref{s1s2} into \eqref{c2} we find
$c_2=c_0$, and consequently, $s_1(p)=s_2(p)$.
Thus, \eqref{h1} of this case reads
\begin{equation}
h(u,\t{u},p)=s_0(p)+s_1(p)u+s_1(p)\t{u},
\label{hh1}
\end{equation}
and \eqref{s1s2} reads
\[2s_1(r)=c_1s_0(r)+c_0s_1^{-1}(r).\]
We note that \eqref{hh1} is already discussed in \cite{A08}.
After making $u\to u-1/c_1$ in \eqref{hh1} we can consider
\[h(u,\t{u},p)=-\frac{c_0}{c_1 s_1(p)}+s_1(p)u+s_1(p)\t{u},
\]
which, however, leads to Case B.2.
So, nothing new is contributed in this case.

\subsection{Multidimensional  consistency:  $h$  given in \eqref{h2}}\label{sec-3-3A}

For $h(u,\t{u},p)$ given in \eqref{h2}, we suppose $s_0(p),s_3(p)\neq0$
and they are not constants simultaneously, so that we can keep the freedom of $p, q$.
By the same manner, from the coefficient of $\t{u}\h{u}^3$ in $ \b{\h{\t{u}}}=\t{\h{\b{u}}}$ , we find
\begin{equation}
 s_3^2(p)s_0^2(r)-s_3^2(r)s_0^2(p)+s_3^2(r)s_0^2(q)-s_3^2(q)s_0^2(r)+s_3^2(q)s_0^2(p)-s_3^2(p)s_0^2(q)=0.
 \label{u2-CAC}
\end{equation}

If $s_3(p)$ is a constant, setting $s_3(p)=1, s_0(p)=p\delta$, we have $h(u,\t{u},p)=u\t{u}+p\delta$, \eqref{u2} is  H3($\delta$)
and \eqref{U2} is  H3($-\delta$).

If $s_0(p)$ is a constant,  by setting $s_0(p)=1, s_3(p)=p\delta$, it comes out  that $h(u,\t{u},p)=\delta pu\t{u}+1$.
By transformation $u\to u^{-1}$, \eqref{u2} also reaches  H3($\delta$).

If neither $s_0(p)$ or $s_3(p)$ is a constant, $\partial_p\partial_q\eqref{u2-CAC}$ yields
\[\frac{ (s_0^2(p))'}{(s_3^2(p))'} =\frac{ (s_0^2(q))'}{(s_3^2(q))'}=c_1,\]
 which leads to $s_0^2(p)=c_1s_3^2(p)+c_2$ with nonzero constant $c_1$.
When $c_2=0$,  it yields $s_0(p)=\delta s_3(p)$. Taking  $u\to (-\delta)^{1/2}u$ and $s_3(p)=1/p$,
we can write equation \eqref{u2} as
\begin{equation}\label{new}
   q^2(u\t{u}-1)(\h{u}\h{\t{u}}-1)=p^2(u\h{u}-1)(\t{u}\h{\t{u}}-1),
\end{equation}
and the corresponding \eqref{U2} is  H3($1$)  after transformation $U\to U^{-1}$.
When $c_2\neq 0$, we can scale $c_1$ to be $1$ by taking  $u\to (c_1)^{1/4}u$.
Then, setting $c_2=1, s_3(p)=\frac{-p}{\sqrt{1-p^2}}, s_0(p)=\frac{1}{\sqrt{1-p^2}}$, equation \eqref{u2} reduces to A2
 and  \eqref{U2} reduces to A2($\sqrt{1-p^2}, \sqrt{1-q^2}$).

We note that \eqref{new} is not a new equation.
It is related to Q1($0;p^2,q^2$) by $u\to u^{(-1)^{n+m}}$.

As a conclusion we reach Theorem \ref{class}.

\section{Proof of Theorem \ref{t-q1}}\label{sec-B}

According to the BT \eqref{BTQ1-1} and assumption \eqref{U} and \eqref{vu},
we can assume $v_i$ have the following special form
\begin{equation}
v_1=(-1)^{\theta_1}\displaystyle{\sum_{i=0}^N}f_{N-i}x^i,
~v_2=(-1)^{\theta_2}\displaystyle{\sum_{i=0}^N}g_{N-i}x^i
\label{vi}
\end{equation}
with constants $f_i, g_i$ to be determined,
where $\theta_i$ can be arbitrary functions of $n,m$.
 First we have the following.
 
\begin{lemma}\label{lem1}
With $U$ defined in \eqref{U} and $v_1, v_2$ defined above, when $N\geq1$, we have
\begin{equation}\label{v}
    v_1=\displaystyle{\sum_{i=0}^N}f_{N-i}x^i, ~~v_2=\displaystyle{\sum_{i=0}^N}g_{N-i}x^i.
\end{equation}
\end{lemma}

\begin{proof}
 Substituting \eqref{U} and \eqref{vi} into system \eqref{BTQ1-1} and \eqref{v12}, from the coefficient of
 the leading term $x^{2N}$ (if $N\geq 1)$ in \eqref{BTQ1-1} we find
\begin{equation}\label{f0g0}
    f_0^2=c_0^2p,\ \ g_0^2=c_0^2q,
\end{equation}
which means $\frac{f_0^2}{g_0^2}=\frac{p}{q}$, and  from the coefficient of $x^{N} $ in \eqref{v12} we find
\begin{equation*}
   f_0\left((-1)^{\h{\theta}_1}-(-1)^{\theta_1}\right)=g_0\left((-1)^{\t{\theta}_2}-(-1)^{\theta_2}\right).
\end{equation*}
Consequently we have
\begin{equation*}
    \left((-1)^{\h{\theta}_1}-(-1)^{\theta_1}\right)^2p=\left((-1)^{\t{\theta}_2}-(-1)^{\theta_2}\right)^2q.
\end{equation*}
Since $p, q$ are independent constants, it follows that $\theta_1=\theta_1(n), \theta_2=\theta_2(m)$.
Consequently, from the coefficient of $x^{N-1} (N\geq1)$ in \eqref{v12} we have
$    (-1)^{\theta_1}f_0b=(-1)^{\theta_2}g_0a$,
which leads to the fact that $\theta_1$ and $\theta_2$ are constants.
 Noticing that $\theta_1$ and $\theta_2$ can be absorbed into $f_i, g_i$, we can assume $\theta_1=\theta_2=0$ without loss of generality.
 Thus \eqref{vi} becomes \eqref{v}.
\end{proof}

\begin{lemma}\label{lem2}
With $U$ defined in \eqref{U} and $v_1, v_2$ defined in \eqref{v}, when $N \geq 1$, the allowed values of $N$ are only 1,2.
\end{lemma}
\begin{proof} Analyzing the coefficient of $x^{N-1} (N\geq1)$ in \eqref{v12}, we obtain $ f_0b=g_0a$,
which leads to $\frac{a^2}{b^2}=\frac{p}{q}$ in light of \eqref{f0g0}.
So we have $p=\theta a^2, q=\theta b^2$ with constant $\theta$.
$\theta$ can be scaled to be 1 using $\delta$, therefore we  set
\begin{equation}\label{pq}
    p=a^2,\quad q=b^2,
\end{equation}
and $f_0=tc_0a,~ g_0=tc_0b,~ t^2=1$ from \eqref{f0g0}.
One can always take $t$ to be 1 as system \eqref{BTQ1-1} remains invariant  under $U\to  -U$. Consequently we have
\begin{equation}\label{f0g0c0}
   f_0=c_0a, \ \ g_0=c_0b.
\end{equation}
Substituting \eqref{v}, \eqref{pq} and \eqref{f0g0c0} into the coefficient of $x^{2N-1} (N\geq1)$ in \eqref{BTQ1-1}, we can work out
\begin{equation}\label{f1g1}
    f_1=\frac{2c_1+c_0aN}{2}a,\ \ g_1=\frac{2c_1+c_0bN}{2}b.
\end{equation}
Then substituting \eqref{pq}-\eqref{f1g1} into  the coefficient of $x^{N-2} (N\geq2)$ in \eqref{v12}, we find  it varnishes.
Next, analyzing the coefficient of $x^{2N-2} (N\geq2)$ in \eqref{BTQ1-1}, we have
\begin{eqnarray}
\begin{aligned}\label{f2g2}
& &f_2=\frac{a}{8}\left(8c_2+4ac_1(N-1)+a^2c_0(N^2-2N)\right),\\
& &g_2=\frac{b}{8}\left(8c_2+4bc_1(N-1)+b^2c_0(N^2-2N)\right).
\end{aligned}
\end{eqnarray}
Substituting \eqref{pq}-\eqref{f2g2} into  the coefficient of $x^{N-3} (N\geq3)$ in \eqref{v12}, we obtain
\begin{equation*}
    abc_0(a^2-b^2)(N^2+2N)=0,
\end{equation*}
which admits zero option if  $N\geq3$.
Therefore all possible choices of $N$ can only be $1, 2$.
\end{proof}

\end{appendix}

\end{document}